\documentclass[aps,prc,twocolumn,groupedaddress]{revtex4-2}

\usepackage{dcolumn}% Align table columns on decimal point
\usepackage{braket}
\usepackage{amsmath} %数学公式包，没有这个包敲公式可能出错
\usepackage{amsthm} % 数学定理包，提供写证明过程的环境的
\usepackage{graphicx} % Required to insert images
\usepackage{amsfonts}%解决dAlembert operator打不出的问题
\usepackage{float}     % 也是和图片设置有关的
\usepackage{physics} % 物理包，不知道好不好用，想打ket才找到这个的
\usepackage{subcaption}

\usepackage{xcolor}

\theoremstyle{definition}
\newtheorem{lemma}{Lemma}
\newtheorem{corollary}{Corollary}
\newtheorem{Theorem}{Theorem}

\begin{document}

\title{New angular momentum coupling method based on Wigner rotation theory}

\author{Junchao Guo}
\email[]{ guojunchao0704@sjtu.edu.cn }
\affiliation{School of Physics and Astronomy, Shanghai Jiao Tong University, Shanghai 200240, China}

\author{ Yang Sun }
\email[]{ sunyang@sjtu.edu.cn }\affiliation{School of Physics and Astronomy, Shanghai Jiao Tong University, Shanghai 200240, China}

\date{\today}

\begin{abstract}
We present a new method for constructing the total angular momentum of many-nucleon states. We find that the restrictions imposed by the fermion antisymmetry on the total state are fully absorbed into the single-$j$ space when the broken rotational symmetry of the product state is restored by angular momentum projection. For different $j$-shells, any total angular momentum that obeys the selection rule is allowed, just as for non-identical particles. The method based on this reorganization is conceptually different from the traditional $J$ and $m$ schemes and may help to improve the efficiency of angular momentum coupling in nuclear many-body calculations.
\end{abstract}

\maketitle

\section{Introduction}

Many-body theories describing atomic nuclei respect rotational symmetry. The Hamiltonian commutes with the rotation operator as well as its generator, i.e. the angular momentum operator. However, this may not be the case for the initial wave functions in the nuclear shell model. Even if the single-particle wave functions are generated from a spherically symmetric potential, the many-body basis (Slater determinants) are product states, which are not eigenstates of angular momentum. Although the mathematical method (the so-called Coefficients of Fractional Parentage, CFP) for construction of
good angular momentum in the many-body basis (i.e. the $J$-scheme) is well developed \cite{Talmi1963}, this method, with increasing number of active nucleons, soon becomes extremely complex in practical applications. Consequently, the $J$-scheme has been shelved and is not much used in shell model calculations today.

In the popular $m$-scheme \cite{Talmi1963}, the tedious angular momentum coupling using CFP is abandoned. It allows admixture of angular momenta $J$ with $J\geq |M|$, so that the basis is no longer irreducible under rotation transformation. Working with the scalar quantity $M$ is pleasant; however, the price one has to pay is the breaking of rotational symmetry in the manybody basis. When the Hamiltonian is diagonalized, the obtained eigenstates contain all possible $J$'s simultaneously. This means that, to cover all relevant $J$ representations in the $m$-scheme, one works with an inefficient basis expressed in a large configuration space. As the name, Large Scale Shell Model (LSSM), for models using the $m$-scheme suggests, one usually has to diagonalize a huge Hamiltonian matrix up to a dimension of $10^9$ for medium-mass nuclei \cite{Caurier2005} (record of $10^{10}$ or larger in dimension reported later by Mizusaki {\it et al.} \cite{Mizusaki2010}, Brown and Rae \cite{Brown2014}, Shimizu \cite{Shimizu2020}), which is obviously the main reason that this method is inapplicable to heavier nuclei.
 
The problem becomes more serious if we are not only interested in a small number of low-lying states close to the ground state, but also in dense highly-excited states, which are often involved in reaction calculations, $\beta$-decay with high $Q$ values, and nuclear fission calculations. In such cases, spin distribution \cite{Egidy2009,Grimes2016,Stetch2014} for certain excitations is a long-standing unsolved problem.

On the other hand, the theory of angular momentum projection (AMP) \cite{Guidry2022} is a well-known many-body technique for restoring broken rotational symmetry for various reasons \cite{Ring-Schuck,Sun2016}. AMP uses quantum mechanical terms to describe the rotational motion and angular momentum of many-body states, which should be consistent with the construction elements of the shell model basis in a sense. However, the possibility of using AMP to build many-body states with good angular momentum for the spherical shell model does not seem to have attracted much attention.

In this paper, we propose a new method to calculate total angular momentum for many-nucleon states of arbitrary spherical $j$-shells using AMP theory. We discuss the foundations of the theory and show rigorous mathematical proofs of angular momentum coupling in fermion systems. Through examples, we demonstrate how our theory changes the idea of the traditional $J$- and $m$-schemes that are widely adopted for spherical shell model calculations. The basic elements of AMP theory and some examples of angular momentum calculation are given in the Appendix.

\section{ Angular momentum projection}

Let us start our discussion with the angular momentum projection (AMP) operator \cite{Guidry2022} 
\begin{equation} \label{AMP}
\begin{aligned}
\hat P^J_{MK} & 
& = \frac{2J+1}{8\pi^2}\int d\Omega\ D_{MK}^{J}(\alpha, \beta, \gamma) \hat R(\alpha,\beta,\gamma) ,
\end{aligned}
\end{equation}
where $\hat R$ is the rotation operator and $D_{MK}^{J}$ the Wigner $D$-function \cite{AMbook}. More about AMP can be found in Appendix A. For now, it is sufficient to know that $\hat P^J$, when applied to  a product state, extracts $J$ (total angular momentum) components from such the rotational-symmetry violated state. 

Suppose in a Hilbert space, there is a complete set of manybody states $\{\ket{\phi_1}, \ket{\phi_2}\dots\ket{\phi_n}\}$, which are not eigenstates of angular momentum. We assume they are eigenstates of $J_z$ and have the same $z$-projection quantum number $M$. Formally, we can construct angular momentum eigenstates as a linear combination of $\{\ket{\phi_1}, \ket{\phi_2}\dots\ket{\phi_n}\}$,
\begin{equation}
\ket{\varphi_i} = \sum_j A_{ij}\ket{\phi_j}.\nonumber
\end{equation}
Since $\ket{\varphi_i}$ is an eigenstate of angular momentum, it is also an eigenstate of the AMP operator.
%\begin{equation}
%\hat P^{J_i}_{MM}\ket{\varphi_i} = %\ket{\varphi_i}.
%\end{equation}
This means that we can find angular momentum eigenstates with good $J$ by diagonalizing the matrix of the AMP operator sandwiched by the $M$-states,
\begin{equation}\label{Norm}
\bra{\phi_i}\hat P^J_{MM}\ket{\phi_j}.
\end{equation}
We call Eq. (\ref{Norm})  Norm matrix \cite{Hara-Sun}, which carries all information about angular momentum for $J$-mixed states $\{\ket{\phi_i}\}$.

This is the basic logic of using angular momentum projection to get angular momentum eigenstates. The key is to calculate matrix elements of the projection operator.

\section{ A warm-up example of AMP: a single-$j$ case}

Single particle states in a spherical basis are labeled by angular momentum quantum numbers $j$ and $m$. Consider a single $j$-shell, $j=\frac{7}{2}$, with $2j+1=8$ single particle states. Now suppose one puts four fermions in this shell. The product states,
\[
\ket{\phi_a} = a_{jm_1}^\dagger a_{jm_2}^\dagger a_{jm_3}^\dagger a_{jm_4}^\dagger \ket{0},
\]
form a complete and orthogonal basis in the Hilbert space, with the dimension $\dbinom{8}{4}=70$. Note that these product states are eigenstates of $J_z$ labeled by $M$, but not  eigenstates of $J$. In other words, rotational symmetry is violated in $\ket{\phi_a}$.

To restore the rotational symmetry, one applies the projection operator, $\hat P^J$, onto the product states. This requires the calculation of matrix elements
\begin{equation}\label{single-j-Matrix}
\bra{\phi_a} \hat P_{M_a M_b}^J\ket{\phi_b},
\end{equation}
where $\ket{\phi_b}$ is another product state in the Hilbert space. According to the integral form of the projection operator in \textcolor{red}{Eq. (\ref{AMP}) } (see also Appendix A), the problem is reduced to the calculation of matrix elements of the rotation operator, $e^{-i\beta J_y}$, in this product basis. The result has a rather pleasant form,
\begin{equation}\label{det}
\bra{\phi_a}e^{-i\beta J_y}\ket{\phi_b} = \operatorname{det} (S^j),
\end{equation}
where $S^j$ with $j={7\over 2}$ is the submatrix constructed from the Wigner $d$-matrix, $d^{j=7/2}$, by choosing the rows corresponding to the $m$ values in $\ket{\phi_a}$ and columns corresponding to the $m$ values in $\ket{\phi_b}$. This result is both compact in theory and efficient in real calculations. One may also understand this result from the perspective of the Grassmann algebra \cite{Grassmann}. 

As a minimal example with $j={3\over 2}$, we take $\ket{\phi_a} = a_{\frac{3}{2}\frac{3}{2}}^\dagger a_{\frac{3}{2}\frac{1}{2}}^\dagger \ket{0}$ and $\ket{\phi_b} = a_{\frac{3}{2}\frac{-3}{2}}^\dagger a_{\frac{3}{2}\frac{-1}{2}}^\dagger \ket{0}$. The corresponding submatrix $S^{3\over 2}$ for this example is:
\[
\large
\begin{bmatrix}
d^{\frac{3}{2}}_{\frac{3}{2}\frac{-3}{2}} & \quad d^{\frac{3}{2}}_{\frac{3}{2}\frac{-1}{2}} \vspace{0.5cm} \\ 
d^{\frac{3}{2}}_{\frac{1}{2}\frac{-3}{2}}  & \quad d^{\frac{3}{2}}_{\frac{1}{2}\frac{-1}{2}} 
\end{bmatrix}.
\]

As another example to illustrate the efficiency of our method, let us consider a special case - a fully-occupied $j$-shell:
\[
\ket{\phi}=\prod_{m=-j}^{j}a_{jm}^{\dagger}\ket{0}.
\]
According to Eq. (\ref{det}), the expectation value of the rotation operator is 
\[
\bra{\phi}e^{-i\beta J_y}\ket{\phi}=\operatorname{det}\left(d^j(\beta)\right),
\]
which is exactly the $j$-representation of the rotation operator. As an element of the SO(3) group, it has the property of $\operatorname{det}d^j(\beta)=1$, so that we know immediately that $\bra{\phi}e^{-i\beta J_y}\ket{\phi}=1$, which shows that this state is a $J = 0$ state.

\section{ AMP for manybody states of multiple $j$-shells }

We now present the core idea of the present study: constructing angular-momentum states of any multiple $j$-shells via angular momentum projection. With mathematical proofs, we first introduce a lemma and two corollaries, and then a theorem.

\begin{lemma}
Suppose there are two different $j$-shells, $j_1$ and $j_2$. By defining
\begin{equation}
\begin{aligned}
\phi_1^\dagger &= a_{j_1, m_1}^\dagger a_{j_1, m_2}^\dagger\cdots a_{j_1, m_s}^\dagger, \nonumber\\
\phi_2^\dagger &= a_{j_2, k_1}^\dagger a_{j_2, k_2}^\dagger \cdots a_{j_2, k_t}^\dagger, \nonumber
\end{aligned}
\end{equation}
we write $\ket{\phi_1} = \phi_1^\dagger\ket{0}$ as a general $s$-particle state in $j_1$-shell and $\ket{\phi_2}$ a $t$-particle state in $j_2$-shell. Similarly, we write $\ket{\phi_1^\prime}$ and $\ket{\phi_2^\prime}$ as another $s$-particle state in $j_1$ and $t$-particle state in $j_2$, respectively. 
For total product states,
\begin{equation}
\begin{aligned}
\ket{\phi} &= a_{j_1, m_1}^\dagger a_{j_1, m_2}^\dagger\cdots a_{j_1, m_s}^\dagger a_{j_2, k_1}^\dagger a_{j_2, k_2}^\dagger \cdots a_{j_2, k_t}^\dagger\ket{0} \nonumber\\
&= \phi_1^\dagger \phi_2^\dagger \ket{0} \nonumber
\end{aligned}
\end{equation}
and
\begin{equation}
\ket{\phi^\prime} =  {\phi_1^\prime}^\dagger {\phi_2^\prime}^\dagger \ket{0}, \nonumber
\end{equation}
the matrix of the rotation operator $e^{-i\beta J_y}$ in this two-$j$ Hilbert space can be decomposed into
\begin{equation}
\bra{\phi^\prime}e^{-i\beta J_y}\ket{\phi} = \bra{\phi_1^\prime} e^{-i\beta J_y}\ket{ \phi_1} \bra{\phi_2^\prime} e^{-i\beta J_y}\ket{ \phi_2} .
\end{equation}

\end{lemma}

\begin{proof}

As seen in the previous section, the matrix elements of $e^{-i\beta J_y}$ for product states in a single $j$-shell is equal to the determinant of matrix $S^j$. 

Now there are two different $j$-shells. As shown in Appendix A, in the second quantization formalism, the fermion creation operator transforms under rotation within the subspace of a given $j$. Namely, when rotating $a_{j_1m_1}^\dagger$, it can only transform into another $m$-state $a_{j_1m_2}^\dagger$ belonging to $j_1$, not any $m$-state of the other shell of $j_2$. In the language of group theory, for each $j$, there is a distinct irreducible representation (irrep) of SO(3). Rotation operations occur within a given irrep but do not transform between irreps. 

Therefore, the matrix $S$, with the matrix elements of $e^{-i\beta J_y}$ in the two-$j$ manybody bases,
\begin{equation}
\bra{\phi^\prime} e^{-i\beta J_y}\ket{\phi}  = \operatorname{det}(S), \nonumber
\end{equation}
must be of a block-diagonal form composed of $S^{j_1}$ and $S^{j_2}$
\begin{equation}\label{blockM}
S = \begin{bmatrix}
 S^{j_1} & \\ 
  & S^{j_2}
\end{bmatrix},
\end{equation}
where $S^{j_1}$ ($S^{j_2}$) is determined by $\ket{\phi_1}$ and $\ket{\phi_1^\prime}$ ($\ket{\phi_2}$ and $\ket{\phi_2^\prime}$) as discussed in section III.

The determinant of a block-diagonal matrix in Eq. (\ref{blockM}) is simply the product of determinants of the two blocks, namely
\begin{equation}\label{eq7}
\begin{aligned}
\bra{\phi^\prime}e^{-i\beta J_y}\ket{\phi}  &= \operatorname{det}(S) \\
&= \operatorname{det}(S^{j_1})\operatorname{det}(S^{j_2}) \\
&=  \bra{\phi_1^\prime}e^{-i\beta J_y}\ket{\phi_1} \bra{\phi_2^\prime}e^{-i\beta J_y}\ket{\phi_2} .
\end{aligned}
\end{equation}

\end{proof}

From Lemma 1 and the proof, it is straightforward to obtain the following two corollaries.
 
\begin{corollary}
If two $j$-shells have different additional quantum numbers, $\nu_1$ or $\nu_2$, making them orthogonal, then the rotation can not mix these two $j$-shells. Even if the two shells have the same $j$, the expectation value can still be factored. These additional quantum numbers include but are not limited to radial quantum number $n$, orbital angular momentum $l$, parity $\pi$, isospin $\tau$.
\end{corollary}

\begin{corollary}
For cases of more $j$-shells than two, the same result holds.
\end{corollary}

With Lemma 1 and the two corollaries, we can now prove the following theorem.

\begin{Theorem}
Following Lemma 1 and Corollary 1, for any two $j$-shells with $j_1\ne j_2$, or $j_1= j_2$ but $\nu_1 \ne \nu_2$, the matrix of the projection operator for the  total product states can be written as 
\begin{eqnarray}\label{Theorem1}
 \bra{\phi^\prime} \hat P^J_{M^\prime M} \ket{\phi}  &&= \sum_{ J_1, J_2} ( \bra{\phi_1^\prime } \hat P_{M_1^\prime M_1}^{J_1} \ket{\phi_1} \bra{\phi_2^\prime} \hat P_{M_2^\prime M_2}^{J_2} \ket{\phi_2} \nonumber \\
&& \bra{J_1M_1^\prime J_2M_2^\prime}\ket{JM^\prime} \bra{J_1M_1 J_2M_2}\ket{JM}  ).
\end{eqnarray}
\end{Theorem}

\begin{proof}
Denote $J_z\ket{\phi_1} = M_1\ket{\phi_1}$, and for each product state in the discussion. Calculate expectation values of the projection operator:
\[
\bra{\phi^\prime} \hat P^J_{M^\prime M} \ket{\phi}  = \dfrac{2J+1}{2} \int_0^{\pi}\sin\beta d\beta \ d^J_{M^\prime M}\bra{\phi^\prime}e^{-i\beta J_y}\ket{\phi} 
\]
Using Lemma 1, we write
\begin{eqnarray}\label{totalphi}
\bra{\phi^\prime} \hat P^J_{M^\prime M} \ket{\phi} &&  = \dfrac{2J+1}{2} \int_0^{\pi}\sin\beta d\beta ~( d^J_{M^\prime M}\times \nonumber \\
&& \bra{\phi_1^\prime}e^{-i\beta J_y}\ket{\phi_1} \bra{\phi_2^\prime}e^{-i\beta J_y}\ket{\phi_2} )
\end{eqnarray}
Plug the identity operator  (see Appendix A) 
\begin{equation}
 I = \sum_{J M} \hat P^J_{MM} = \sum_{\nu J M}\ket{\nu J M}\bra{\nu J M} \nonumber
 \end{equation}
into $\bra{\phi_1^\prime}e^{-i\beta J_y}\ket{\phi_1} $ twice, we get
\begin{eqnarray}\label{phi1}
 && \bra{\phi_1^\prime}e^{-i\beta J_y}\ket{\phi_1} \nonumber \\
% = && \sum_{\nu_1J_1 M}\sum_{\nu_2J_2 K} \bra{\phi_1^\prime}\ket{\nu_1J_1M}\bra{\nu_1J_1M}e^{-i\beta J_y}\ket{\nu_2J_2 K}\bra{\nu_2J_2 K}\ket{\phi_1} \nonumber \\
 = && \sum_{\nu_1 J_1 M_1 M_1^\prime} \bra{\phi_1^\prime}\ket{\nu_1J_1M_1^\prime}\bra{\nu_1J_1M_1^\prime}e^{-i\beta J_y}\ket{\nu_1J_1 M_1}\bra{\nu_1J_1 M_1}\ket{\phi_1} \nonumber %\\ = && \sum_{\nu_1 J_1 M_1 M_1^\prime}  \bra{\phi_1^\prime}\ket{\nu_1J_1M_1^\prime}\bra{\nu_1J_1M_1}\ket{\phi_1} \bra{\nu_1J_1M_1^\prime}e^{-i\beta J_y}\ket{\nu_1J_1M_1} . \nonumber 
 \end{eqnarray}
The middle term in the above summation is nothing but the definition of the small-$d$ function, $d^J_{M^\prime M}$. Considering that $\sum_\nu \ket{\nu J M^\prime}\bra{\nu J M}$ is the angular momentum projection operator $\hat P_{M^\prime M}^J $, we obtain
\begin{equation}
\bra{\phi_1^\prime}e^{-i\beta J_y}\ket{\phi_1}  =  \sum_{ J_1 } \bra{\phi_1^\prime}\hat P_{M_1^\prime M_1}^{J_1} \ket{\phi_1} d^{J_1} _{M_1^\prime M_1}.
\end{equation}
Similarly,
\begin{equation}\label{phi2}
%\[
\bra{\phi_2^\prime}e^{-i\beta J_y}\ket{\phi_2} =  \sum_{ J_2 } \bra{\phi_2^\prime}\hat P_{M_2^\prime M_2}^{J_2} \ket{\phi_2} d^{J_2}_{M_2^\prime M_2}.
\end{equation}
%\]

Putting Eqs. (\ref{phi1}) and (\ref{phi2}) into Eq. (\ref{totalphi}), we get
\begin{eqnarray}
&& \bra{\phi^\prime} \hat P^J_{M^\prime M} \ket{\phi} \nonumber \\
= && \dfrac{2J+1}{2} \int_0^{\pi}\sin\beta d\beta \ d^J_{M^\prime M}  \bra{\phi_1^\prime}e^{-i\beta J_y}\ket{\phi_1} \bra{\phi_2^\prime}e^{-i\beta J_y}\ket{\phi_2} \nonumber \\
= && \dfrac{2J+1}{2} \sum_{ J_1, J_2}  \bra{\phi_1^\prime} \hat P_{M_1^\prime M_1}^{J_1} \ket{\phi_1} \bra{\phi_2^\prime} \hat P_{M_2^\prime M_2}^{J_2} \ket{\phi_2}\times \nonumber \\
&& \hspace{2cm} \int_0^{\pi} \sin\beta d\beta \ d^J_{M^\prime M} d^{J_1}_{M_1^\prime M_1} d^{J_2}_{M_2^\prime M_2} \nonumber
\end{eqnarray}
Using the property of Wigner $d$-function
\begin{equation}
\begin{aligned}
 & \dfrac{2J+1}{2}  \int_0^{\pi} \sin\beta d\beta \ d^J_{MK} d^{J_1}_{M_1K_1} d^{J_2}_{M_2K_2} \\
 = & \left\langle J_1 M_1 J_2 M_2 \mid J M\right\rangle \left\langle J_1 K_1 J_2 K_2 \mid J K\right\rangle , \nonumber
 \end{aligned}
\end{equation}
we finally get
\begin{eqnarray}
 \bra{\phi^\prime} \hat P^J_{M^\prime M} \ket{\phi}  &&= \sum_{ J_1, J_2} ( \bra{\phi_1^\prime } \hat P_{M_1^\prime M_1}^{J_1} \ket{\phi_1} \bra{\phi_2^\prime} \hat P_{M_2^\prime M_2}^{J_2} \ket{\phi_2} \nonumber \\
&& \bra{J_1M_1^\prime J_2M_2^\prime}\ket{JM^\prime} \bra{J_1M_1 J_2M_2}\ket{JM}  )\nonumber
\end{eqnarray}

\end{proof}

\section{ discussions }

Among the set of quantum numbers that uniquely represent a nucleus, those related to angular momentum are $j$ and $m$ in single-particle space or $J$ and $M$ in many-particle states. The remaining degrees of freedom, denoted by $\nu$, are involved to determine exchange properties of identical fermions. Conventional angular momentum coupling methods usually treat nucleons as non-identical particles, leaving the imposition of identical-particle restriction until the last step \cite{Talmi1963}. For example, in the $m$-scheme, a large number of redundant configurations are freely generated, but they are not the physical states required by antisymmetry and must be removed at additional computational cost.

We treat the angular momentum problem differently. In Theorem 1, we make full use of the properties of the projection operator to distinguish the role of single-$j$ and multi-$j$ shells in angular momentum coupling. Fermionic antisymmetry requires that the restriction on the allowed total states be completely absorbed into the single $j$ space, but this is only the initial step of the process. For different $j$ shells, as shown in Eq.  (\ref{Theorem1}), any total angular momentum that obeys the coupling selection rule is allowed, as if the coupling was performed on non-identical particles. Clearly, a great advantage of our method is that the final angular momentum coupling is only performed after the physical angular momentum states in each $j$-shell are prepared.

We can list some additional outstanding advantages of our method in practical applications.

1.
Theorem 1 is a recurrence relation. If there are more than two $j$-shells, one only needs to apply Eq. (\ref{Theorem1}) successively. For the construction of many-nucleon states with good angular momentum, our method can in principle be applied to a large model space with multiple $j$-shells without numerical difficulties. 

2.
Theorem 1 may replace the traditional $J$-scheme \cite{Talmi1963} for constructing total angular-momentum states. Abandoning the CFP method completely, we first prepare eigenstates of angular momentum in single $j$-shells by diagonalizing the projection operator matrix of Eq. (\ref{single-j-Matrix}). Then 
total angular-momentum states for  multiple $j$-shell spaces can 
be easily obtained by the usual angular momentum coupling via 
CG-coefficients, as Eq. (\ref{Theorem1}) shows. Our method may help to reduce computational intensity for calculations such as those in Ref. \cite{Dytrych2016}.

3.
Lemma 1 may be used in selecting $J$-states in the last step of the $m$-scheme \cite{Talmi1963}. Suppose in a usual LSSM calculation \cite{BrownNotes}, after diagonalization of the Hamiltonian 
\begin{equation}
\hat H\ket{\phi_i} = E\ket{\phi_i},
\end{equation}
each $\ket{\phi_i}$ is generally a superposition of multi-shell product state, from which many energy eigenstates with high degeneracy can be obtained. To determine eigenstates of angular momentum out of $\ket{\phi_i}$, the usual approach is to calculate the expectation value of $\hat J^2$, from which one can construct a different form of AMP operator \cite{Comments}. In practical applications using this operator, one ends up with a set of algebraic equations \cite{Ring-Schuck}. As described in Sec. II, we calculate the projection operator matrix essentially for single-$j$ spaces with the projection operator defined in (\ref{AMP}). Using Lemma 1, the calculation may be greatly simplified and more efficient. Corollary 2 shows that it is easy to generalize Lemma 1 to the cases with an arbitrary number of $j$-shells.

\section{ Conclusions }

During the late 1920s and early 1930s, Eugene Wigner made profound contributions to the study of angular momentum and its projection in the framework of quantum mechanics. Wigner’s work, by introducing the $D$-function, provided an elegant mathematical tool to analyze the projection of quantum states in a rotationally invariant way and to express how these states transform under rotation. Inspired by Wigner's work, the Generator Coordinate Method (GCM) was established \cite{Hill1953,Griffin1957}, and has become a powerful many-body method for nuclear structure calculations \cite{Reinhard1987,Chen2013,Yao2014,Egido2016}.  It was discussed \cite{Sun2016} that the angular momentum projection can be viewed as a special application of the GCM by making use of the completeness of the Wigner $D$-function to expand the generator function in the GCM. In the 1970s and 1980s, people began to use the theory of angular momentum projection to restore the rotation symmetry that was broken in mean-field calculations in deformed nuclei \cite{Ring-Schuck,Hara1979,Schmid1987}, and established the Projected Shell Model \cite{Sun2016,Hara1991,Hara-Sun}. Since then, angular momentum projection has become the mainstream method in all kinds of beyond-mean-field calculations \cite{Sun1994a,Sun1994b,Sun1998,Otsuka2001,Rodriguez2005,
Sheihk2008,Chen2008,Wang2014,Johnson2017,Bally2019,Dao2022,Kaneko2023,Wang2024}. However, to the best of our knowledge, angular momentum projection, as a practical tool for angular momentum coupling of manybody states, has not been considered for shell model based on spherical basis.

The present research is based entirely on Wigner's theory of angular momentum and continues to explore the power of the angular momentum projection operator. The new method we proposed changes the traditional way of thinking for construction of angular momentum of fermionic quantum many-body states. Although the discussion in this article is concentrated on $J$ and $m$ schemes, this method may be generally applied. It is applicable in any many-nucleon problems (including reaction and fission theories) where (1) the problem requires good total angular momentum and (2) the spherical model basis involves multiple $j$ shells. More applications of the method proposed in this article are being studied and the efficiency of our method is to be verified, but we expect that our method may help to solve some standing problems related to angular momentum in shell model and other applications.

\begin{acknowledgments}
We thank B. A. Brown, F.-Q. Chen, C. W. Johnson, T. Mizusaki, and L.-J. Wang for helpful comments and suggestions. This work is supported by the National Natural Science Foundation
of China (Grant No. 12235003). %JCG acknowledges financial support of special fund xxx from Shanghai Jiaotong University. 
\end{acknowledgments}

\appendix
\section{ Basic elements of AMP}

In this appendix, a basic introduction to the AMP method and its connection with the usual angular momentum coupling of Fermions are given.

The Wigner $D$-function is the matrix presentation of the rotation operator $\hat R$ in the basis of angular momentum eigenstates $\{\ket{jm}\}$ specified by Euler angles,
\begin{equation}
\begin{aligned}
& D_{m^{\prime} m}^j(\alpha, \beta, \gamma)=\left\langle j m^{\prime}|\hat R(\alpha, \beta, \gamma)| j m\right\rangle , \\
& \quad \hat R(\alpha, \beta, \gamma)=e^{-i \alpha J_z} e^{-i \beta J_y} e^{-i \gamma J_z} .
\end{aligned}
\end{equation}
Here $\alpha$, $\beta$, $\gamma$ are the Euler angles, sometimes denoted with one symbol, $\Omega$. The non-trivial part of $D$ is the small-$d$ function
\begin{equation}
\begin{aligned}
& D_{m^{\prime} m}^j=e^{-i m^{\prime} \alpha} d_{m^{\prime} m}^j(\beta) e^{-i m \gamma} , \\
& d_{m^{\prime} m}^{j}(\beta)=\bra{ j m^{\prime}}e^{-i \beta J_y}\ket{ j m }.
\end{aligned}
\end{equation}

Including other quantum numbers denoted by $\nu$, which are unchanged by rotation, we can write
\begin{equation}
 \left\langle\nu^\prime j m^{\prime}|\hat R(\alpha, \beta, \gamma)|\nu j m\right\rangle = \delta_{\nu^\prime \nu} D_{m^{\prime} m}^j(\alpha, \beta, \gamma) .
\end{equation}
It should be noted that $\ket{\nu j m }$ is a generic notation, which may either be a single-particle or a many-body state. 

In the second quantization formalism, the creation operator of Fermions transforms under rotation in the following way
\begin{equation}
\hat R~a_{jm}^\dagger ~\hat R^\dagger = D^j_{m^{\prime} m} a_{jm^{\prime}}^\dagger .
\end{equation}
The orthogonality relation of $D$ is 
\begin{equation}
\int d\Omega D_{m^{\prime} k^{\prime}}^{j^{\prime}}(\Omega)^* D_{m k}^j(\Omega)=\frac{8 \pi^2}{2 j+1} \delta_{m^{\prime} m} \delta_{k^{\prime} k} \delta_{j^{\prime} j} ,
\end{equation}
where $d\Omega=\sin\beta ~d\alpha ~d\beta ~d\gamma$. From Eqs. (A3) and (A5), we can obtain the following relation:
\begin{equation}
\begin{aligned}
& \frac{2j+1}{8\pi^2}\int \sin\beta ~d\alpha ~d\beta ~d\gamma\ D_{m k}^{j}(\alpha, \beta, \gamma)^* \hat R(\alpha,\beta,\gamma) \\
=&\sum_{\nu} \ket{\nu jm}\bra{\nu jk} ,
\end{aligned}
\end{equation}
which is defined as the angular momentum projection operator $\hat P^j_{mk}$. 

To get a physical picture of this projection operator, think of the following special case. Suppose we have a product state $\ket{\phi_K}$ which is an eigenstate of $J_z$, i.e. $J_z\ket{\phi_K}=K\ket{\phi_K}$. Decompose it in terms of angular momentum eigenstates:
\begin{equation}
\ket{\phi_K} = \sum_{\nu J}\ket{\nu J K}\bra{\nu J K} \ket{\phi_K}  =\sum_J \hat P_{KK}^J \ket{\phi_K} .
\end{equation}
This means that the projector $\hat P_{KK}^J$ extracts the angular momentum $J$ components from $\ket{\phi_K}$. By calculating the expectation value of $\hat P_{KK}^J$, we thus obtain the angular momentum distribution, the probability of finding $J$ in $\ket{\phi_K}$
\begin{equation}
\bra{\phi_K}\hat P_{KK}^J\ket{\phi_K} = \sum_{\nu}|\bra{\nu J K \ket{\phi_K}}^2 = |c^J|^2.
\end{equation}

To understand how the projection method works, let's look at a simple example:
\[
\ket{\phi}_{\frac{1}{2}} = \dfrac{1}{\sqrt{3}} \left( \ket{\nu_1 \dfrac{1}{2} \dfrac{1}{2}} +  \ket{\nu_2 \dfrac{3}{2} \dfrac{1}{2}} +  \ket{\nu_3 \dfrac{3}{2} \dfrac{1}{2}}  \right) ,
\]

\[
\begin{aligned}
\hat P_{\frac{1}{2}\frac{1}{2}}^{\frac{1}{2}} \ket{\phi}_{\frac{1}{2}} &=  \dfrac{1}{\sqrt{3}}  \ket{\nu_1 \dfrac{1}{2} \dfrac{1}{2}}, \\
\hat P_{\frac{1}{2}\frac{1}{2}}^{\frac{3}{2}} \ket{\phi}_{\frac{1}{2}} &=  \dfrac{1}{\sqrt{3}}  \left( \ket{\nu_2 \dfrac{3}{2} \dfrac{1}{2}} + \ket{\nu_3 \dfrac{3}{2} \dfrac{1}{2}} \right) ,
\end{aligned} 
\]
which leads to
\[
|c^{\frac{1}{2}}|^2=\bra{\phi}\hat P_{\frac{1}{2}\frac{1}{2}}^{\frac{1}{2}}  \ket{\phi} = \dfrac{1}{3}, \quad |c^{\frac{3}{2}}|^2=\bra{\phi}\hat P_{\frac{1}{2}\frac{1}{2}}^{\frac{3}{2}}  \ket{\phi} = \dfrac{2}{3} .
\]

The above AMP method may seem trivial for single particle states. However, it shows its powerful effect on manybody states. The key is perhaps that the projection operator uses the right language -- the rotation operator -- which rotates the many-body state (with good angular momentum) while respecting the exchange properties of fermions.  In contrast, the traditional angular momentum coupling method usually works for non-identical particles, which is inconsistent with the second quantization form of fermionic states.

\section{ Derivation of Equations (\ref{det}) and (\ref{blockM})}

We first give the derivation for Eq. (\ref{det}) through an example of 2-particle states, which can be easily  extended to cases with any number of particles. 

Suppose $\ket{\phi_b}$ is 2-particle state of a $j$-shell, expressed as
\[
\ket{\phi_b} =a_{j m_1}^{\dagger}a_{j m_2}^{\dagger}\ket{0} ,
\]
and $\ket{\phi_a}$ as
\[
\ket{\phi_a} =a_{j m_3}^{\dagger}a_{j m_4}^{\dagger}\ket{0} .
\]
Applying the rotation operator on the ket state $\ket{\phi_b}$, it becomes
\[
e^{-i\beta J_y}\ket{\phi_b}=\sum_{m_1'm_2'}d^{j}_{m_1' m_1}d^{j}_{m_2' m_2}a_{j m'_1}^\dagger a_{j m'_2}^\dagger \ket{0} .
\]
The inner product of $\ket{\phi_a}$ and the rotated $\ket{\phi_b}$ is
\[
\bra{\phi_a}e^{-i\beta J_y}\ket{\phi_b} =\sum_{m_1'm_2'} \bra{0}a_{j m_4}a_{j m_3}d^{j}_{m_1' m_1}d^{j}_{m_2' m_2}a_{j m'_1}^\dagger a_{j m'_2}^\dagger \ket{0} .
\] 

The only non-vanishing terms in the above summation are precisely the determinant of the $2\times 2$ matrix constructed from the Wigner $d$-function for the $j$-shell
\[
S^{j} = \begin{bmatrix}
d^{j}_{m_3 m_1} & \ d^{j}_{m_3 m_2} \vspace{0.5cm} \\ 
d^{j}_{m_4 m_1}  & \ d^{j}_{m_4 m_2} 
\end{bmatrix}.
\]

To give the proof for the block-diagonal form of matrix in Eq. (\ref{blockM}), let us suppose $\ket{\phi_1} = a^\dagger_{j_1m_1}a^\dagger_{j_1m_2}\ket{0}$, $\ket{\phi_2} = a^\dagger_{j_2m_3}a^\dagger_{j_2m_4}\ket{0}$,  $\ket{\phi_1'} = a^\dagger_{j_1m'_1}a^\dagger_{j_1m'_2}\ket{0}$ and $\ket{\phi'_2} = a^\dagger_{j'_2m'_3}a^\dagger_{j'_2m'_4}\ket{0}$. 
We then write the details of Eq. (\ref{eq7}) as
\[
\begin{aligned}
\bra{\phi'}e^{-i\beta J_y}\ket{\phi} = \sum_{M_1M_2M_3M_4} &\bra{0} a_{j_2m'_4}a_{j_2m'_3}a_{j_1m'_2}a_{j_1m'_1}\\&d^{j_1}_{M_1 m_1}d^{j_1}_{M_2 m_2}d^{j_2}_{M_3 m_3}d^{j_2}_{M_4 m_4}\\&a^\dagger_{j_1M_1}a^\dagger_{j_1M_2}a^\dagger_{j_2M_3}a^\dagger_{j_2M_4} \ket{0} .
\end{aligned}
\]
Notice that the fermion operator transforms under rotation only within the subspace of a given $j$. For example, if $M_1=m'_3$ or $m'_4$, the term in the summation must vanish because $j_1\neq j_2$. In other words, the summation over $M_1$ and $M_2$ (subindex of $j_1$) runs independently with that over $M_3$ and $M_4$ (subindex of $j_2$). Therefore, the resulting matrix for the multiple $j$-shells in Lemma 1 must be of block-diagonal form 
\[
S = 
\begin{bmatrix}
S^{j_1} &  \\ 
  & S^{j_2} 
\end{bmatrix} = 
\begin{bmatrix}
d^{j_1}_{m'_1 m_1} & d^{j_1}_{m'_1 m_2} & & \vspace{0.2cm} \\ 
d^{j_1}_{m'_2 m_1}  & d^{j_1}_{m'_2 m_2} & & \vspace{0.2cm} \\
 & & d^{j_2}_{m'_3 m_3} & d^{j_2}_{m'_3 m_4} \vspace{0.2cm}  \\
 & & d^{j_2}_{m'_4 m_3} & d^{j_2}_{m'_4 m_4} 
\end{bmatrix} ,
\]
specified by $j$.

\section{ New derivations of some related propositions }

As shown before, our theory can greatly simplify the angular momentum construction of nuclear manybody states. This appendix provides two examples showing how the method can be applied in practical calculations. Both textbook examples are known propositions in basic nuclear theory, usually derived via standard angular momentum coupling methods. The derivations here are new. A more detailed discussion of the examples is given in \cite{Guo-Thesis}, and the application of the method to solving existing spin-distribution problems will be published  elsewhere. 

\subsection{ Particle-hole symmetry in angular momentum }

For any single-$j$ shell, there exists a particle-hole mirror symmetry with respect to the mid-shell. In other words, knowledge on angular momentum of hole states is sufficient for knowing that of the corresponding particle states.

\begin{proof}
Let's denote a fully-occupied single $j$-shell state as $\ket{\phi}$. Then the two-hole (2-h) state can be written as
\[
\ket{\varphi}=a_{jm_1}a_{jm_2}\ket{\phi} .
\]

As we already know $\hat R\ket{\phi}=\ket{\phi}$, we can derive the expectation value of rotation of this 2-h state as
\[
\begin{aligned}
\bra{\varphi}e^{-i\beta J_y}\ket{\varphi}&=\bra{\phi}a_{jm_2}^\dagger a_{jm_1}^\dagger e^{-i\beta J_y}a_{jm_1}a_{jm_2}\ket{\phi}\\
&= \bra{\phi}a_{jm_2}^\dagger a_{jm_1}^\dagger d^j_{m_im_1}d^j_{m_km_2}a_{jm_i}a_{jm_k}e^{-i\beta J_y}\ket{\phi}\\
&=\bra{\phi}a_{jm_2}^\dagger a_{jm_1}^\dagger d^j_{m_im_1}d^j_{m_km_2}a_{jm_i}a_{jm_k}\ket{\phi}\\
&=\operatorname{det}( S^j ) ,
\end{aligned}
\]
where $S^j$ is equal to that of the two-particle (2-p) state $a_{jm_1}^\dagger a_{jm_2}^\dagger\ket{0}$. 

This shows that under rotation, the 2-h state behaves exactly the same as the corresponding 2-p state. As a result, they must have the same angular momentum. This conclusion holds for any $n$ number of particles-holes, suggesting that there exists a particle-hole mirror symmetry in any $j$ shell with respect to the mid-shell.

\end{proof}

\subsection{ Pairing and angular momentum }

The angular momentum of a fully-paired state must be even. In other words,  a fully-paired state only allows eigenstates of even angular momentum.
\begin{proof}
Consider the product state for a single $j$-shell occupied by four particles: 
\begin{equation}
\ket{\phi}=a_{j_1m_1}^{\dagger}a_{j_1m_2}^{\dagger}a_{j_1m_3}^{\dagger}a_{j_1m_4}^{\dagger}\ket{0} ,
\end{equation}
where $m_1+m_2=0$, $m_3+m_4=0$, so that the total $M=\sum_i m_i =0$. 

The angular momentum distribution of $\ket{\phi}$ is 
\begin{equation}\label{AMPdis}
\bra{\phi}\hat P^J_{00}\ket{\phi}=\dfrac{2I+1}{2}\int_0^{\pi}\sin\beta d\beta ~d^J_{00}(\beta)\bra{\phi}e^{-i\beta J_y}\ket{\phi},
\end{equation}
where the expectation value of the rotation operator is 
\begin{equation}
\bra{\phi}e^{-i\beta J_y}\ket{\phi}=\operatorname{det}\left( S^{j} \right)
\end{equation}
Because of the paired condition, this expectation has a symmetry with $\beta=\frac{\pi}{2}$. To see this, use the property of the Wigner small-$d$ function:
\begin{equation}
d_{m^{\prime}, m}^j(\pi-\beta)=(-1)^{j+m^{\prime}} d_{m^{\prime},-m}^j(\beta).
\end{equation}
Due to this property, $S^{j}(\pi - \beta)$ can be obtained from the matrix $S^{j}(\beta)$ by swapping the first and second  columns, the third and fourth columns, and multiplying each $i$-th row by $(-1)^{j + m_i}$. This operation can be extended to cases of any $n$ pairs. For the resulting $2n \times 2n$ matrix, performing $n$ column swaps and multiplying each row by a factor leads to the determinant being multiplied by a factor of $(-1)^n \cdot (-1)^{2nj} = (-1)^{(2j + 1)n}$. For fermions, $j$ is a half-odd-integer, so that $2j + 1$ is an even integer, which ensures the determinant to be an even function that is symmetric about $\beta = \frac{\pi}{2}$. 

Note that $d^J_{00}(\beta)=P_J(\cos(\beta))$, and the Legendre polynomial is an even (odd) function with respect to $\beta=\dfrac{\pi}{2}$ when $J$ is even (odd). Therefore, when $J$ is odd, the integral in Eq. (\ref{AMPdis}) must be zero, indicating that in the fully-paired scenarios, total angular momentum of the state must be even.

This proof can be generalized to the multi-shell case, where Lemma 1 ensures the conclusion.

\end{proof}

%\bibliography{ paper }

\begin{thebibliography}{}

\bibitem{Talmi1963} A. de-Shalit and I. Talmi, {\it Nuclear Shell Theory} (Academic Press, New York and London, 1963).

\bibitem{Caurier2005} E. Caurier, G. Mart\'inez-Pinedo, F. Nowacki, A. Poves, and A. P. Zuker, Rev. Mod. Phys. {\bf 77}(2), 427 (2005).

\bibitem{Mizusaki2010} T. Mizusaki, K. Kaneko, M. Homma, and T. Sakurai, Phys. Rev. C {\bf 82}, 024310 (2010).

\bibitem{Brown2014} B. A. Brown and W. D. M. Rae,
Nucl. Data Sheets {\bf 120}, 115 (2014).

\bibitem{Shimizu2020} N. Shimizu,
https://doi.org/10.48550/arXiv.1310.5431

\bibitem{Egidy2009} T. von Egidy and D. Bucurescu, Phys. Rev. C {\bf 80}, 054301 (2009).

\bibitem{Grimes2016} S. M. Grimes, A. V. Voinov, and T. N. Massey, Phys. Rev. C {\bf 94}, 014308 (2016).

\bibitem{Stetch2014} I. Stetch, P. Talou, T. Kawano, and M. Jandel, Phys. Rev. C {\bf 90}, 024617 (2014).

\bibitem{Guidry2022} M. Guidry and Y. Sun, {\it Symmetry, Broken Symmetry, and Topology in Modern Physics: A First Course} (Cambridge University Press, Cambridge, 2022).

\bibitem{Ring-Schuck} P. Ring and P. Schuck, {\it The Nuclear Many-Body Problem} (Springer–Verlag, New York, 1980).

\bibitem{Sun2016} Y. Sun, Phys. Scr. {\bf 91}, 043005 (2016).

\bibitem{AMbook} D. A. Varshalovich, A. N. Moskalev, and V. K. Khersonskii, {\it Quantum Theory of Angular Momentum} (World Scientific, Singapore, 1988).

\bibitem{Hara-Sun} K. Hara and Y. Sun, Int. J. Mod. Phys. E {\bf 4}, 637 (1995).

\bibitem{Grassmann} Ch. Doran and A. Lasenby, {\it Geometric Algebra for Physicists} (Cambridge University Press, Cambridge, 2003).

\bibitem{Dytrych2016} T. Dytrych, P. Maris, K. D. Launey, J. P. Draayer, J. P. Vary, D. Langr, E. Saule, M. A. Caprio, U. Catalyurek, M. Sosonkina, Comput. Phys. Commun., {\bf 207}, 202 (2016).

\bibitem{BrownNotes} B. A. Brown, {\it Lecture Notes in Nuclear Structure Physics}, https://people.frib.msu.edu/~brown/2025/brown-lecture-notes-2025.pdf

\bibitem{Comments} We thank B. A. Brown and T. Mizusaki for bringing this to our attention.

%\bibitem{Wigner1931} E. Wigner, J. Am. Chem. Soc. {\bf 53}(7), 2706 (1931).

%$\bibitem{Wigner1939} E. Wigner, J. Math. Phys. {\bf 1}(1), 44 (1939).

%\bibitem{Wigner1959} E. Wigner, {\it Group Theory and Its Application to the Quantum Mechanics of Atomic Spectra} (Academic Press, New York, 1959).

\bibitem{Hill1953} D. L. Hill and J. A. Wheeler, Phys. Rev. {\bf 89}, 1102 (1953).

\bibitem{Griffin1957} J. J. Griffin and J. A. Wheeler, Phys. Rev. {\bf 108}, 311 (1957).

\bibitem{Reinhard1987} P.-G. Reinhard and K. Goeke, Rep. Prog. Phys. {\bf 50}, 1 (1987).

\bibitem{Chen2013} F.-Q. Chen, Y. Sun, and P. Ring, Phys. Rev. C {\bf 88},  014315 (2013).

\bibitem{Yao2014} J. M. Yao, K. Hagino, Z. P. Li, J. Meng, and P. Ring, Phys. Rev. C {\bf 89}, 054306 (2014).

\bibitem{Egido2016} J. L. Egido, Phys. Scr. {\bf 91}, 073003 (2016).

\bibitem{Hara1979} K. Hara and S. Iwasaki, Nucl. Phys. A {\bf 332}, 61 (1979); 
A {\bf 348}, 200 (1980); A {\bf 430}, 175 (1984).

\bibitem{Schmid1987} K. W. Schmid and F. Gr\"ummer, Rep. Prog. Phys. {\bf 50}, 731 (1987).

\bibitem{Hara1991} K. Hara and Y. Sun, Nucl. Phys. A {\bf 529}, 445 (1991); A {\bf 531}, 221 (1991); A {\bf 537}, 77 (1992).
 
\bibitem{Sun1994a} Y. Sun and J. L. Egido, Phys. Rev. C {\bf 50}, 1893 (1994).

\bibitem{Sun1994b} Y. Sun, L. M. Robledo, and J. L. Egido, Nucl. Phys. A {\bf 570}, 305 (1994).

\bibitem{Sun1998} Y. Sun, C.-L. Wu, K. Bhatt, M. Guidry, and D. H. Feng, Phys. Rev. Lett. {\bf 80}, 672 (1998).

\bibitem{Otsuka2001} T. Otsuka, M. Honma, T. Mizusaki, N. Shimizu, and Y. Utsuno, Prog. Part. Nucl. Phys. {\bf 47}, 319 (2001).

\bibitem{Rodriguez2005} T. R. Rodr\'iguez, J. L. Egido, L. M. Robledo, and R. Rodriguez-Guzman, Phys. Rev. C {\bf 71}, 044313 (2005).

\bibitem{Sheihk2008} J. A. Sheihk, G. H. Bhat, Y. Sun, G. B. Vakil, and R. Palit, Phys. Rev. C {\bf77}, 034313 (2008).

\bibitem{Chen2008} Y.-S. Chen, Y. Sun, and Z.-C. Gao, Phys. Rev. C {\bf77},  061305(R) (2008).

\bibitem{Wang2014} L.-J. Wang, F.-Q. Chen, T. Mizusaki, M. Oi, and Y. Sun, Phys. Rev. C {\bf 90},  011303(R) (2014).

\bibitem{Johnson2017} C. W. Johnson and K. D. O’Mara, Phys. Rev. C {\bf 96}, 064304 (2017).

\bibitem{Bally2019} B. Bally, A. S\'anchez-Fern\'andez, and T. R. Rodr\'iguez, Phys.
Rev. C {\bf 100}, 044308 (2019).

\bibitem{Dao2022} D. D. Dao and F. Nowacki, Phys. Rev. C {\bf 105}, 054314 (2022).

\bibitem{Kaneko2023} K. Kaneko, Y. Sun, N. Shimizu, and T. Mizusaki, Phys. Rev. Lett. {\bf 130}, 052501 (2023).

\bibitem{Wang2024} Y. P. Wang, Y. K. Wang, F. F. Xu, P. W. Zhao, and J. Meng, Phys. Rev. Lett. {\bf 132}, 232501 (2024).

\bibitem{Guo-Thesis} J.-C. Guo, Doctoral thesis to be submitted to Shanghai Jiaotong University, September 2025.

\end{thebibliography}

\end{document}